\documentclass[preprint,12pt,3p]{elsarticle}

\RequirePackage[OT1]{fontenc}
\RequirePackage{amsthm,amsmath}
\RequirePackage{natbib}
\RequirePackage[colorlinks,citecolor=blue,urlcolor=blue]{hyperref}

\theoremstyle{plain}

\usepackage{float}
\usepackage[T1]{fontenc}
\usepackage{graphicx}
\usepackage{floatflt}
\usepackage{bm}
\usepackage{rotating}
\usepackage{multirow}
\usepackage{booktabs}
\usepackage{xspace}
\usepackage{color}
\usepackage{braket}
\usepackage{algpseudocode}
\usepackage{tabularx}
\usepackage{pdfcomment}
\usepackage{lineno}
\usepackage{amsmath}
\usepackage{amssymb}
\usepackage{amsthm}
\usepackage{bm}
\usepackage{rotating}
\usepackage{multirow}
\usepackage{booktabs}
\usepackage{placeins}
\usepackage{wrapfig}
\usepackage{graphicx}
\graphicspath{{./figures/}}
\usepackage{enumerate}
\usepackage{mathtools}
\usepackage{cleveref}
\usepackage{subcaption}
\usepackage{verbatim}

\usepackage{lipsum}
\usepackage{todonotes}
\usepackage{standalone}
\usepackage{hyperref}
\usepackage{url}

\usepackage{algorithm}
\usepackage{algpseudocode}

\newcommand{\indep}{\rotatebox[origin=c]{90}{$\models$}}

\DeclarePairedDelimiterX{\ExpArg}[1]{[}{]}{#1}

\theoremstyle{plain}
\newtheorem{theorem}{Theorem}[section]
\newtheorem{lemma}[theorem]{Lemma}

\newtheorem{proposition}[theorem]{Proposition}

\theoremstyle{definition}
\newtheorem{definition}{Definition}[section]

\theoremstyle{remark}

\definecolor{lightgray}{gray}{0.95}

\newcommand{\Trace }[1]{\mbox{}{\bf{Tr}}\left(#1\right)}

\newcommand{\NormS}[1]{\mbox{}\left\|#1\right\|^2}

\newcommand{\setlinespacing}[1]%
           {\setlength{\baselineskip}{#1 \defbaselineskip}}

\newcommand{\BlockDiagk}[1]{\mbox{}\left(%
\begin{array}{cc}
  \Sigma_{k} & \bf{0} \\
  \bf{0} &  \Sigma_{\rho-k}\\
\end{array}\right)}

\newcommand{\BlockDiagkk}[1]{\mbox{}\left(%
\begin{array}{cc}
  \Sigma_{k} & \bf{0} \\
  \bf{0} & \bf{0} \\
\end{array}\right)}

\newcommand{\BlockDiagkrk}[1]{\mbox{}\left(%
\begin{array}{cc}
  \bf{0} & \bf{0} \\
  \bf{0} & \Sigma_{\rho-k} \\
\end{array}\right)}

\newcommand{\BlockDiagkkh}[1]{\mbox{}\left(%
\begin{array}{c}
  \Sigma_{k} \\
  \bf{0} \\
\end{array}\right)}

\newcommand{\BlockDiagkrkh}[1]{\mbox{}\left(%
\begin{array}{c}
  \bf{0} \\
  \Sigma_{\rho-k} \\
\end{array}\right)}
\newenvironment{Proof}{\noindent {\em Proof:}}

\long\def\killtext#1{}

\errorcontextlines\maxdimen

\makeatletter
    \newcommand*{\algrule}[1][\algorithmicindent]{\makebox[#1][l]{\hspace*{.5em}\thealgruleextra\vrule height \thealgruleheight depth \thealgruledepth}}%
\newcommand*{\thealgruleextra}{}
\newcommand*{\thealgruleheight}{.75\baselineskip}
\newcommand*{\thealgruledepth}{.25\baselineskip}

\newcount\ALG@printindent@tempcnta
\def\ALG@printindent{%
    \ifnum \theALG@nested>0
        \ifx\ALG@text\ALG@x@notext
        \else
            \unskip
            \addvspace{-1pt}
            \ALG@printindent@tempcnta=1
            \loop
                \algrule[\csname ALG@ind@\the\ALG@printindent@tempcnta\endcsname]%
                \advance \ALG@printindent@tempcnta 1
            \ifnum \ALG@printindent@tempcnta<\numexpr\theALG@nested+1\relax
            \repeat
        \fi
    \fi
    }%
\usepackage{etoolbox}
\patchcmd{\ALG@doentity}{\noindent\hskip\ALG@tlm}{\ALG@printindent}{}{\errmessage{failed to patch}}
\makeatother

\newbox\statebox
\newcommand{\myState}[1]{%
    \setbox\statebox=\vbox{#1}%
    \edef\thealgruleheight{\dimexpr \the\ht\statebox+1pt\relax}%
    \edef\thealgruledepth{\dimexpr \the\dp\statebox+1pt\relax}%
    \ifdim\thealgruleheight<.75\baselineskip
        \def\thealgruleheight{\dimexpr .75\baselineskip+1pt\relax}%
    \fi
    \ifdim\thealgruledepth<.25\baselineskip
        \def\thealgruledepth{\dimexpr .25\baselineskip+1pt\relax}%
    \fi
    \State #1%
    \def\thealgruleheight{\dimexpr .75\baselineskip+1pt\relax}%
    \def\thealgruledepth{\dimexpr .25\baselineskip+1pt\relax}%
}

\usepackage{amssymb}

\journal{arXiv}

\begin{document}

\begin{frontmatter}

\title{Combinatorics of Distance Covariance:\\ Inclusion-Minimal Maximizers of Quasi-Concave Set Functions for Diverse Variable Selection}
\author[label1,label2]{Praneeth Vepakomma\corref{cor1}}
\address[label1]{Rutgers University}
\address[label2]{Motorola Solutions\fnref{label4}}

\cortext[cor1]{corresponding author}

\ead{praneeth@scarletmail.rutgers.edu}

\author[label5]{Yulia Kempner}
\address[label5]{Holon Institute of Technology}
\ead{yuliak@hit.ac.il}

\begin{abstract}
In this paper we show that the negative sample distance covariance function is a quasi-concave set function of samples of random variables that are not statistically independent. We use these properties to propose greedy algorithms to combinatorially optimize some diversity (low statistical dependence) promoting functions of distance covariance. Our greedy algorithm obtains all the inclusion-minimal maximizers of this diversity promoting objective. Inclusion-minimal maximizers are multiple solution sets of globally optimal maximizers that are not a proper subset of any other maximizing set in the solution set. We present results upon applying this approach to obtain diverse features (covariates/variables/predictors) in a feature selection setting for regression (or classification) problems. We also combine our diverse feature selection algorithm with a distance covariance based relevant feature selection algorithm of \cite{WahbaFS} to produce subsets of covariates that are both relevant yet ordered in non-increasing levels of diversity of these subsets.
\end{abstract}


\end{frontmatter}



\section{Introduction}
\subsection{The classical problem of variable selection:}
The problem of "variable selection" also known as "feature selection" or "covariate selection" is a prominent problem in statistics and machine learning. The goal in here is to be able to choose an optimal subset of covariates in a regression or classification setting that would perform optimally with respect to the out-of-sample prediction or classification accuracy when the chosen subset is used to predict (or classify) one or more real-valued response variables (in regression) or one or more categorical variables (in classification). There have been an umpteen number of techniques developed for this problem under a broadly varying spectrum of assumptions.
\subsection{The more recent problem of diverse variable selection:}
Traditional feature selection algorithms have the primary goal of finding the best feature subset that is relevant to a regression or classification task. More recently, there has been a strong focus on not just the above mentioned goal of relevant feature selection but also on selecting a "small" subset of "diverse" features. Diversity is
useful for several reasons such as interpretability, robustness to noise and in some cases to cater to reduction of real-life costs of costly feature acquisition for guiding feature engineering to decide on what other features could be acquired etc.
\subsubsection{Some existing work on diversification:}The authors in \cite{divfs} provide a solution in the specific case of linear regression through a formulation where a diversity promoting sub-modular regularizer is added to the standard linear regression problem. In this setting the solution is obtained by greedy algorithms that optimize a submodular function based objective. Although this is an interesting approach, we'd like to point that this approach restricts the regression model to be linear unlike it being generalized to any regression (non-linear and linear) models. Another important issue with this approach is that their solution is not globally optimal but is instead an approximation with a well-known (in submodular optimization) $1-\frac{1}{e}$ styled guarantee of $\frac{(1-e^{-b.\gamma(U,k)}).OPT}{c}$, where $OPT$ is the optimal solution, $\gamma(U,k)$ is a function of the solution subset of features $U$ obtained through their algorithm and it's cardinality $k$ (i.e., the number of features in $U$). $b, c$ are algorithm dependent constants depending on the specific choice of \textit{algorithm} out of  multiple algorithms that they propose. Note that $(1-\frac{1}{e})$ is approximately equal to $63\%$.\subsubsection{Existing work on diversification with mutual-information:} Another popular approach is \cite{mrmr} which is based on measuring diversity and relevancy through functions of mutual-information. Their solution approximately optimizes their proposed objective as obtaining a global solution would require $O(n^{|S|})$ search operations where $n$ is the number of samples and $|S|$ is the cardinality of the number of features required to be selected by the algorithm. This can be a prohibitively large number in the case of many practical datasets and required $|S|$.

\subsection{Advantages of our proposed algorithms:}
 The technique proposed in below sections of our paper has two major advantages:
 \begin{enumerate}
     \item Our solution to our proposed diversity encouraging objective function is globally optimal with no approximation error unlike the $1-\frac{1}{e}$ styled approximate solution provided by \cite{divfs} or the unquantified approximation error provided by \cite{mrmr}. We also propose an approach that is completely devoid of any parameters and provide a global solution to our proposed formulation. That said, we do completely recognize that the objective function proposed in our technique varies from the objective functions proposed in existing techniques.

 \item Another advantage of our approach is that it is independent of the choice of regression (linear/non-linear) or classification (linear/non-linear) model to be used unlike the work by \cite{divfs} which focusses only on linear regression.

 \item Our approach can directly be used for diversified feature selection in both cases of univariate or multivariate (vector-valued) responses (in regression) or multi-label (in classification) without modifying our proposed objective function or algorithm while the approach in \cite{divfs} does not seem to extend trivially beyond the univariate response case in linear regression without modifying their regularized objective function or algorithmic routines.
 \end{enumerate}

Prior to getting into the crux of our proposed theoretical results and algorithmic implications, we'd like to note that in theory our approach can be explicitly parametrized by a trade-off parameter to control the trade-off between relevancy and diversity of features selected. Such tuning of trade-offs is not the main focus of this paper. The previously proposed approach using spectral regularization \cite{divfs} does parametrize this through regularization parameters that weigh the submodular regularizer appropriately.
\section{Problem Formulation:}
  In this paper we cover the following three problems: 

 \begin{enumerate}[Problem I:]

 \item \textbf{Diverse Feature Selection}\\ The goal here is to find a subset of features that have the least statistical dependence amongst each other. This implies that the selected features would be diverse.\par

   \item \textbf{All-Relevant Feature Selection}\\The goal here is to find a subset of features that are most statistically dependent on a response variable.


     \item \textbf{Diverse and Relevant Feature Selection}\\The goal here is to find a subset of features that are more statistically dependent on a response variable while also being less statistically dependent amongst each other.
 \end{enumerate}
We present a greedy-algorithm with exactly optimal solutions in this paper for our formulated objective to solve Problem I. We point to an existing  approach for  Problem II and propose simple methodologies for Problem III where the methodologies are based on solutions of Problem I and II. Before we get to the main result of our paper, our suggested two simple methodological approaches for Problem III are:
\begin{enumerate}[(a.)]
\item Controlled approach: In this approach, we first choose a subset of features that "individually" have a statistical dependency \textit{i.e $\geq \alpha \in \mathcal{R}^{+}$}    with response variable and call this subset the controlled set. We then run our algorithm proposed for Problem I for choosing a diverse set of features from this controlled set.
\item Two-stage approach: In this approach, the Problem III could be approached by solving Problem II followed by Problem I or vice-versa.

\subsection{\textbf{Main Result of the paper:}}So to clearly reiterate, our main and most novel contribution of this paper is our proposed algorithm for Problem I.
\end{enumerate}
 \section{Preliminaries:}
In this section we introduce some preliminaries about distance correlation and distance covariance which we extensively use in our paper to build up towards our proposed theoretical results.
\subsection{\textbf{Distance Covariance and Distance Correlation:}}
Distance Correlation \cite{qcf0} is a measure of nonlinear statistical dependencies between random vectors of arbitrary dimensions. We describe below distance covariance $\mathbb{\nu}^2(\mathbf{x},\mathbf{y})$ between random variables $\mathbf{x} \in \mathbb{R}^d$ and $\mathbf{y} \in \mathbb{R}^m$ with finite first moments is a non-negative number as
	\begin{equation}\label{charac}
		\mathbb{\nu}^2(\mathbf{x},\mathbf{y})=\int_{\mathbb{R}^{d+m}}|f_{\mathbf{x},\mathbf{y}}(t,s)-f_\mathbf{x}(t)f_\mathbf{y}(s)|^2 w(t,s)dtds
	\end{equation}

where $w(t, s)$ is a weight function as defined in \cite{qcf0}, $f_\mathbf{x},f_\mathbf{y}$ are characteristic functions of $\mathbf{x},\mathbf{y}$ and $f_{\mathbf{x},\mathbf{y}}$ is the joint characteristic function.\par The distance covariance is zero if and only if random variables $\mathbf{x}$ and $\mathbf{y}$ are independent. Using the above definition of distance covariance, we have the following expression for Distance Correlation from \cite{qcf0}:\par
 The squared Distance Correlation between random variables $\mathbf{x} \in \mathbb{R}^d$ and $\mathbf{y} \in \mathbb{R}^m$ with finite first moments is a nonnegative number is defined as

\begin{equation} \rho^2(\mathbf{x},\mathbf{y})
	    =\left\{ \begin{array}{cc}
				 \frac{\mathbb{\nu}^2(\mathbf{x},\mathbf{y})}{\sqrt{\mathbb{\nu}^2(\mathbf{x},\mathbf{x})\mathbb{\nu}^2(\mathbf{y},\mathbf{y})}},
				 & \mathbb{\nu}^2(\mathbf{x},\mathbf{x})\mathbb{\nu}^2(\mathbf{y},\mathbf{y})>0.\\
	        0, & \mathbb{\nu}^2(\mathbf{x},\mathbf{x})\mathbb{\nu}^2(\mathbf{y},\mathbf{y})=0.
	        \end{array} \right.
	\end{equation}

The Distance Correlation defined above has the following interesting properties;

\begin{enumerate}
    \item ${\rho}^2(\mathbf{x},\mathbf{y})$	 is applicable for arbitrary dimensions $d$ and $m$ of $\mathbf{x}$ and $\mathbf{y}$ respectively.

    \item ${\rho}^2(\mathbf{x},\mathbf{y})=0$ if and only if $\mathbf{x}$ and $\mathbf{y}$ are independent.

    \item ${\rho}^2(\mathbf{x},\mathbf{y})$ satisfies the relation $0 \leq \rho^2(\mathbf{x},\mathbf{y}) \leq 1$.
\end{enumerate}
 \subsection{\textbf{Sample Distance Covariance and Sample Distance Correlation:}}
 We provide the definition of sample version of distance covariance \cite{qcf0} given samples $\{ (\mathbf{x}_k,\mathbf{y}_k) | k = 1,2,\ldots, n \}$ sampled i.i.d. from joint distribution of random vectors $\mathbf{x} \in \mathbb{R}^d$ and $\mathbf{y} \in \mathbb{R}^m$. To do so, we define two squared Euclidean distance matrices $\mathbf{E}_\mathbf{X}$ and $\mathbf{E}_\mathbf{Y}$,  where each entry $[\mathbf{E}_\mathbf{X}]_{k,l} = \NormS{\mathbf{x}_k-\mathbf{x}_l}$ and $[\mathbf{E}_\mathbf{Y}]_{k,l} = \NormS{\mathbf{y}_k-\mathbf{y}_l}$ with $k,l \in \{ 1,2,\ldots, n\}$. These squared distance matrices are then double-centered by making their row and column sums zero and are denoted as $\widehat{\mathbf{E}}_{\mathbf{X}}, \widehat{\mathbf{Q}}_{\mathbf{X}}$, respectively. So given a double-centering matrix $\mathbf{J}=\mathbf{I}-\frac{1}{n}\mathbf{1}\mathbf{1}^T$, we have $\widehat{\mathbf{E}}_\mathbf{X}=\mathbf{J}\mathbf{E}_\mathbf{X}\mathbf{J}$ and $\widehat{\mathbf{E}}_\mathbf{Y}=\mathbf{J}\mathbf{E}_\mathbf{Y}\mathbf{J}$. The sample distance covariance and sample distance correlation can now be defined as follows:

\begin{definition}{\textbf{Sample Distance Covariance \cite{qcf0}:}}
Given i.i.d samples $\mathcal{X} \times \mathcal{Y} = \{ (\mathbf{x}_k,\mathbf{y}_k) | k = 1,2,3,\ldots, n\}$ and corresponding double centered Euclidean distance matrices $\widehat{\mathbf{E}}_\mathbf{X}$ and $\widehat{\mathbf{E}}_\mathbf{Y}$, then the squared sample distance correlation is defined as,
\label{popDC}\[
    \hat{\mathbb{\nu}}^2(\mathbf{X},\mathbf{Y})=\frac{1}{n^2}\sum_{k,l=1}^{n}[\widehat{\mathbf{E}}_\mathbf{X}]_{k,l}[\widehat{\mathbf{E}}_\mathbf{Y}]_{k,l},
\]	
\end{definition}
Using this, sample distance correlation is given by
\label{sampleDC}
\[
	\hat{\rho}^2(\mathbf{X},\mathbf{Y})
	= \left\{ \begin{array}{cc}
    	 \frac{\mathbf{\hat{\nu}}^2(\mathbf{X},\mathbf{Y})}{\sqrt{\mathbf{\hat{\nu}}^2(\mathbf{X},\mathbf{X})\mathbf{\hat{\nu}}^2(\mathbf{Y},\mathbf{Y})}}, & \mathbf{\hat{\nu}}^2(\mathbf{X},\mathbf{X})\mathbf{\hat{\nu}}^2(\mathbf{Y},\mathbf{Y})>0. \\
    	0, & \mathbf{\hat{\nu}}^2(\mathbf{X},\mathbf{X})\mathbf{\hat{\nu}}^2(\mathbf{Y},\mathbf{Y})=0.
	\end{array}
	\right.
\]

\section{\textbf{Kosorok's Distance Covariance Independence Inequality:}}
If $\mathbf{X,Z} \in \mathbb{R}^p$ and $\mathbf{Y} \in \mathbb{R}^q$ and if and only if $\mathbf{Z} \indep (\mathbf{X},\mathbf{Y})$ then \begin{equation}\nu^2(\mathbf{X}+ \mathbf{Z},\mathbf{Y}) \leq \nu^2(\mathbf{X},\mathbf{Y})\end{equation}
Note that $\indep$ indicates 'statistically independent' in statistical literature. This implies that for each $\mathbf{X,Y,Z}$ that are not pairwise statistically independent (i.e distance covariance between components of any subset of cardinality 2 of $\mathbf{X,Y,Z}$  is positive) then
\begin{equation}\nu^2(\mathbf{X}+ \mathbf{Z},\mathbf{Y}) > \nu^2(\mathbf{X},\mathbf{Y})\end{equation}

 \section{Proof of Kosorok's Distance Covariance Inequality}

The Kosorok's Distance Covariance Independence Inequality was proved in \citep{WahbaFS, DCIneq} and is based on the property of characteristic functions (denoted below by $f$) that

\begin{equation}\label{dcieqns}
    \vert{f_{\mathbf{X + Z},\mathbf{Y}}(t,s) - f_{\mathbf{X + Z}}}(t) f_{\mathbf{Y}}(s)\vert^2 \leq \vert{f_{\mathbf{Z}}}(t) \vert^2 \vert{f_{\mathbf{X },\mathbf{Y}}(t,s) - f_{\mathbf{X }}}(t) f_{\mathbf{Y}}(s)\vert^2
\end{equation} and $\vert{f_{\mathbf{Z}}}(t) \vert^2 \leq 1$. The equation (\ref{dcieqns}) above can be obtained by these facts
\begin{align}
\mid f_{\mathbf{X}+ \mathbf{Z},\mathbf{Y}}(t, s) - f_{\mathbf{X}+\mathbf{Z}}(t)f_{\mathbf{Y}} (s)\vert ^ 2 &= \vert \mathop{\mathbb{E}}e^{it^T(\mathbf{X}+\mathbf{Z})+is^{T} \mathbf{Y}} - \mathop{\mathbb{E}}e^{it^T(\mathbf{X}+\mathbf{Z})}\mathop{\mathbb{E}}e^{is^T\mathbf{Y}} \vert^2\\ \nonumber
      &=\vert \mathop{\mathbb{E}}e^{it^{T}\mathbf{X}+is^{T} \mathbf{Y}} \mathop{\mathbb{E}}e^{it^{T}\mathbf{Z}}-\mathop{\mathbb{E}}e^{it^{T}\mathbf{X}}\mathop{\mathbb{E}}e^{it^{T}\mathbf{Z}}\mathop{\mathbb{E}}e^{it^{T}\mathbf{Y}}\vert^2\\ \nonumber
      &=\vert f_{\mathbf{X},\mathbf{Y}}(t, s)f_{\mathbf{Z}}(t)-f_{\mathbf{X}}(t)f_{\mathbf{Z}}(t)f_{\mathbf{Y}}(s)\vert^2\\ \nonumber
      &= \vert{f_{\mathbf{\mathbf{Z}}}}(t) \vert^2 \vert{f_{\mathbf{X },\mathbf{Y}}(t,s) - f_{\mathbf{X }}}(t) f_{\mathbf{Y}}(s)\vert^2
      \end{align}
      which with implication from $\vert{f_{\mathbf{Z}}}(t) \vert^2 \leq 1$ gives \begin{equation}
          \nu^2({\mathbf{X}+ \mathbf{Z}, \mathbf{Y}}) \leq \nu^2(\mathbf{X},\mathbf{Y})
      \end{equation}
      We know that if $\mathop{\mathbb{E}}\vert \mathbf{X} \vert_d < \infty $, $\mathop{\mathbb{E}}\vert \mathbf{X} + \mathbf{Z}\vert_m < \infty$ and $\mathop{\mathbb{E}}\vert \mathbf{Y} \vert _d < \infty$, then from \cite{qcf0}
$$\lim_{n\to\infty}\nu^2_{n}(\mathbf{X} + \mathbf{Z}, \mathbf{Y} ) = \nu^2(\mathbf{X} + \mathbf{Z}, \mathbf{Y} )$$
and $$\lim_{n\to\infty}\nu^2_{n}(\mathbf{X}, \mathbf{Y} ) = \nu^2(\mathbf{X},\mathbf{Y})$$
Thus, for the sample distance covariance, if $n$ is large enough, we should have $$V^2_{n}(\mathbf{X}+\mathbf{Z}, \mathbf{Y} ) \leq V^2(\mathbf{X},\mathbf{Y})$$ only under the assumption of independence between $(\mathbf{X}, \mathbf{Y} )$ and $\mathbf{Z}$. Note that $\nu_n$ indicates sample distance covariance and $\nu$ indicates population distance covariance.\\
 \par

 \underline{Note}: In the case where considering $(\mathbf{X} \cup \mathbf{Z})$ is of interest, we could use the above theorem by incorporating degenerated random vectors as follows: Suppose $\mathbf{X} \in \mathbb{R}^{p1}$ and $\mathbf{Z} \in \mathbb{R}^{p2}$, then we augment $\mathbf{X}$ and
$\mathbf{Z}$ to be $\tilde{\mathbf{X}} = (\mathbf{X}, \mathbf{0}_{p2})$ and $\tilde{\mathbf{Z}} = (\mathbf{0}_{p1}, \mathbf{Z})$ respectively. $\tilde{\mathbf{X}}$ and $\tilde{\mathbf{Z}}$ are therefore of the same dimension and $\tilde{\mathbf{X}} + \tilde{\mathbf{Z}} = (\mathbf{X}, \mathbf{Z})$. Therefore the $\mathbf{X} \cup \mathbf{Z}$ operation in the context of computing $\hat{\nu}(\mathbf{X}\cup\mathbf{Z},\mathbf{Y})$ with matrices $\mathbf{X},\mathbf{Z},\mathbf{Y}$ is equivalent to appending the columns of $\mathbf{X}$ with the columns of $\mathbf{Z}$ followed by computing the sample-distance covariance between the resulting matrix and $\mathbf{Y}$.  \section{Quasi-Concave Set Functions}

\subsection{Notation and definitions:} We now describe some notation and introduce some definitions that we use through out the paper in the sections below. We use bold faced $\mathbf{X}$ to denote the complete ground set of features/covariates and indexed $X_i$ to denote the $i$'th covariate.  That is we use $i$ indexed subsets like $S_i$ to indicate a singleton (unit cardinality) element of $\mathbf{S}$ labeled by $i$. We denote the response variable in a regression setting with $\mathbf{Y}$.  We denote the set $2^\mathbf{X} \setminus \left\{\phi,\mathbf{X}\right\}$ by $\mathcal{P}^{-X}$ and use  $\setminus$ to denote set difference, i.e $\mathbf{X \setminus Z} = \left\{x:x \in \mathbf{X} \text{ and } x\not \in \mathbf{Z}\right\}$. \par Given a set system $(\mathbf{X},\mathcal{F})$ which is a collection $\mathcal{F}$ of subsets of a ground set $\mathbf{X}$ where $\mathcal{F} \subseteq 2^\mathbf{X}$, we define a quasi-concave set function as given below.

\begin{definition}[\textbf{Quasi-Concave Set Function \cite{Main1},\cite{Mullat}:}]\label{qcvxDef}
A function $F : \mathcal{F} \mapsto \mathbb{R}$ defined on a set system $(\mathbf{X}, \mathcal{F})$ is quasi-concave
if for each $\mathbf{S, T}\in \mathcal{F}$, \begin{equation}\label{mon_func}
F(\mathbf{S} \cap \mathbf{T} ) \geq \min{\{F(\mathbf{S}), F(\mathbf{T})\}}\end{equation}
\end{definition}

\begin{definition}[\textbf{Monotone Linkage Function  \cite{Mullat}:}]

A function $\pi(X_i,\mathbf{Z})$ defined on $\mathbf{Z} \in \mathcal{P}^{-X}  , X_i \in \mathbf{X} \setminus \mathbf{Z}$ is called a monotone linkage function if \begin{equation}\pi(X_i, \mathbf{S}) \geq \pi(X_i, \mathbf{T}), \mathbf{S} \subseteq \mathbf{T}\in \mathcal{F},  \forall X_i \in \mathbf{X} \setminus T\end{equation}
\end{definition}

We'd like to note for the clarity of the reader that $X_i$ is an element while $\mathbf{S},\mathbf{T}$ are sets. Therefore, to make this distinction clear we denote sets in bold-faced font and elements otherwise.

\section{Some Combinatorial Properties of Negative Distance Covariance} We now prove some quasi-concave as well as monotone linkage set function properties of some functions of negative distance covariance.
\begin{theorem}[Quasi-Concave Distance Covariance Set Function Theorem]\label{thm:mvt}


 If we have $\mathbf{S}\cap \mathbf{T} \neq \varnothing \text{ and } \forall \mathbf{S}, \mathbf{T}, \mathbf{Y} \text{ if } \nu^2(\mathbf{S}, \mathbf{T}) > 0 \land \nu^2(\mathbf{S}, \mathbf{Y}) > 0 \land \nu^2(\mathbf{T}, \mathbf{Y}) > 0  \text{ then we have }$ \begin{equation}-\nu^2(\mathbf{S} \cap \mathbf{T}, \mathbf{Y}) \geq min(-\nu^2(\mathbf{S},\mathbf{Y}),-\nu^2(\mathbf{T},\mathbf{Y}))\end{equation}
\end{theorem}

\begin{proof}


If $\mathbf{S} \cap \mathbf{T} = \mathbf{S}$ then since $\mathbf{S} \subseteq \mathbf{T}$


the Kosorok's distance covariance inequality implies \begin{equation}-\nu^2(\mathbf{S},\mathbf{Y}) \geq -\nu^2(\mathbf{T},\mathbf{Y})\end{equation}  Therefore we have
\begin{equation}-\nu^2(\mathbf{S} \cap \mathbf{T}, \mathbf{Y}) \geq min(-\nu^2(\mathbf{S},\mathbf{Y}),-\nu^2(\mathbf{T},\mathbf{Y}))\nonumber\end{equation}
Similarly, if $\mathbf{S} \cap \mathbf{T} = \mathbf{T}$, then since $\mathbf{T} \subseteq \mathbf{S}$
 \begin{equation}-\nu^2(\mathbf{T},\mathbf{Y}) \geq -\nu^2(\mathbf{S},\mathbf{Y})\end{equation}and therefore
\begin{equation}-\nu^2(\mathbf{S} \cap \mathbf{T}, \mathbf{Y}) \geq min(-\nu^2(\mathbf{S},\mathbf{Y}),-\nu^2(\mathbf{T},\mathbf{Y}))\end{equation}

 In the cases of ${\mathbf{S} \cap \mathbf{T}} \subset {\mathbf{S}}$ and  ${\mathbf{S} \cap \mathbf{T}} \subset {\mathbf{T}}$
the Kosorok's distance covariance inequality implies
\begin{equation}-\nu^2(\mathbf{S} \cap \mathbf{T}, \mathbf{Y}) > -\nu^2(\mathbf{S},\mathbf{Y})\end{equation} and
\begin{equation}-\nu^2(\mathbf{S} \cap \mathbf{T}, \mathbf{Y}) > -\nu^2(\mathbf{T},\mathbf{Y})\end{equation}
So \begin{equation}-\nu^2(\mathbf{S} \cap \mathbf{T}, \mathbf{Y}) \geq min(-\nu^2(\mathbf{S},\mathbf{Y}),-\nu^2(\mathbf{T},\mathbf{Y}))\end{equation}

\end{proof}

\subsection{\textbf{A monotone linkage function of distance covariance:}}
\begin{lemma}
The function $\pi(X_i,\mathbf{S})$ of distance covariance defined on $X_i \notin \mathbf{S}$  as \begin{equation}\label{piEqn}
    \underset{X_i \notin \mathbf{S}}{\pi(X_i,\mathbf{S})} = \sum_{\mathbf{S}_j \in \mathbf{S}} -\nu^{2}(X_i,\mathbf{S}_j)
\end{equation}
 is a monotone linkage function \end{lemma}\begin{Proof} For $\mathbf{S} \subseteq \mathbf{T}$ we have \begin{equation}\underset{X_i\notin \mathbf{T}}{\pi(X_i,\mathbf{T})}=\sum_{\mathbf{S}_j \in \mathbf{S}} -\nu_{i}^{2}(X_i,\mathbf{S}_j) -\sum_{\mathbf{T}_j \in \mathbf{T \setminus S}} \nu_{i}^{2}(X_i,\mathbf{T}_j) \leq
 \underset{X_i\notin \mathbf{T}  }
 {\pi(X_i,\mathbf{S})}=\sum_{\mathbf{S}_j \in \mathbf{S}} -\nu_{i}^{2}(X_i,\mathbf{S}_j)\end{equation}We would also like to note that as $\nu(\cdot)$ is a non-negative function the above inequality does hold true.
 \end{Proof}

 \begin{theorem}\cite{Main1}\begin{enumerate}[i)]

The function $M_{\pi}(\mathbf{T}) = \underset{X_i \in \mathbf{X}\setminus \mathbf{T}}{\text{min}}
  \pi(X_i,\mathbf{T})$ is a quasi-concave set function.
\end{enumerate}
\end{theorem}
\begin{Proof}
The proof is in the proof of Assertion 1 in \cite{Main1}
\end{Proof}

\section{\textbf{Diverse Feature Selection:}}We aim to find all the subsets that maximize the function $M_{\pi}(\mathbf{T})$ which result in the solutions which for diverse features.

\begin{equation}\label{minFeqn1}
   \underset{\mathbf{T} \subset \mathbf{X}} {\mathrm{arg\enskip max }} \enskip M_{\pi}(\mathbf{T})
\end{equation}

The above equation (\ref{minFeqn1}) can be written as

\begin{equation}\label{minFeqn3}
     \underset{\mathbf{T} \subset \mathbf{X}} {\mathrm{arg\enskip max }} \enskip \underset{X_i \in \mathbf{X}\setminus \mathbf{T}}{\text{min}}
  \pi(X_i,\mathbf{T})
\end{equation}

This problem does not necessarily have a single, unique solution and hence we aim to find all the subsets that are maximizers of (\ref{minFeqn3}). These are essentially subsets that are each maximally separated from their corresponding nearest neighbor where the notion of nearness to their neighbor is given by (\ref{piEqn}).
 \begin{definition}[\textbf{$ \pi$-series:}]
 We refer to a series $s_{\pi}=(X_{i_1},\ldots, X_{i_N})$ as a $\pi$-series if \begin{equation}\pi({X}_{i_{k+1}}, \bf{\overline{S}_k}) = \underset{X_i \in \mathbf{X}\setminus \mathbf{\overline{S}_k}}{\text{min}}
  \pi(\textnormal{X}_i,\mathbf{\overline{S}_k}) \end{equation}
  
  for any starting set $\bf{\overline{S}_k} = \{X_{i_1},\ldots,X_{i_k}\}, k = {1,\ldots,N-1}$.
 \end{definition}
 Therefore it is a way of greedily populating a series that can start with any first element $\bf{X}_{i_1}$ being the current series, but the subsequent element to be added to the series, must be the element that minimizes the element to current series function of $\pi(\bf{X}_{i_{k+1}},\bf{\overline{S}_k})$ where $\bf{X}_{i_{k+1}}$ is the next element added and $\bf{\overline{S}_k}$ is the current series.
 \begin{definition}[\textbf{$\pi$-cluster}]
A subset $\bf{S}\in \mathcal{P}^{-\mathbf{X}}$ will be
referred to as a $\pi$-cluster if there exists a $\pi$-series, $s_\pi = (X_{i_1},\ldots,X_{i_N})$, such that $\bf{S}$ is a maximizer of $M_{\pi}(\bf{\overline{S}_k})$ over all starting sets $\bf{\overline{S}_k}$ of $s_\pi$.
 \end{definition}

\begin{theorem}\label{Theorem7.1}\cite{Main1}
 If for a $\pi$-series $s_{\pi} = (X_{i_1},X_{i_2},\ldots,i_N)$, a subset $\mathbf{S}\subset \mathbf{X}$ contains $X_{i_1}$, and if $X_{i_{k+1}}$ is the first element in $s_{\pi}$ not contained in $\mathbf{S}$ (for some $k \in  \{1,\ldots, N - 1\}$, then

\begin{equation}M_{\pi}(\mathbf{\overline{S}_k}) \geq  M_{\pi}(\mathbf{S})\end{equation}

where $\mathbf{\overline{S}_k} = \left(X_{i_1},\ldots, X_{i_k}\right )$. In particular, if $\mathbf{S}$ is an inclusion-minimal maximizer of $M_{\pi}$ (with regard
to $\mathcal{P}^{-\mathbf{X}})$, then $\mathbf{S} = \mathbf{\overline{S}_k}$, that is, $\mathbf{S}$ is a $\pi$-cluster.
\end{theorem}
\begin{proof}
 $M_{\pi}(\mathbf{\overline{S}_k}) = \pi(\textnormal{X}_{i_{k+1}},\mathbf{\overline{S}_k})$ by definition. Since $\bf{\overline{S}_k} \subseteq \bf{S}$ we have $\pi(X_{i_{k+1}},\bf{\overline{S}_k}) \geq \pi(\textnormal{X}_{i_{k+1}},\bf{S})$ by monotonicity. To end the proof, note that
$\pi(X_{i_{k+1}}, \bf{S}) \geq M_{\pi}(\mathbf{S})$ because $M_{\pi}(\mathbf{S}) = \underset{X_i \in \mathbf{X}\setminus \mathbf{Z}}{\text{min}}
  \pi(X_i,\mathbf{S}) $ and $X_{i_{k+1}} \notin \bf{S}$
\end{proof}
\begin{proposition}\cite{Main1}
If $\bf{S_1},\bf{S_2} \subset \bf{X}$ are overlapping maximizers of a quasi-concave set function $M_\pi(\bf{S})$ over $\mathcal{P}^{-\bf{X}}$, then $\bf{S_1} \cap \bf{S_2}$ is also a maximizer of $M_\pi(\bf{S})$.
\end{proposition}
\begin{proof}
It directly follows from  (\ref{mon_func}).
\end{proof}This implies that the minimal maximizers of a quasi-convex set function are not overlapping. Moreover, any nonminimal maximizer can be uniquely partitioned into a set of the minimal ones.

\begin{theorem}
Each maximizer of a quasi-concave set function on $\mathcal{P}^{-\bf{X}}$ is a union of its
inclusion-minimal maximizers.
\end{theorem}
\begin{proof}
 Indeed, if $\bf{S^\ast}$ is a maximizer of $M_\pi(\bf{S})$ over $\mathcal{P}^{-\bf{X}}$, then, according to Theorem \ref{Theorem7.1}, for
any $X_{i} \in \bf{S^\ast}$, there exists a minimal maximizer included in $\bf{S^\ast}$ and containing $X_{i}$.
\end{proof}

\subsection{\textbf{Our greedy algorithm for diverse variable selection with distance covariance for solving Problem I}:}

  \begin{algorithm}[H]
   \caption{DiverseMinimalMaximDCoV: Diverse Combinatorial Distance Covariance}
    \begin{algorithmic}[1]
      \Function{=DiverseMinimalMaximDCoV}{$\mathbf{X}$}
\ForAll{$X_i \in \bf{X}$} \\\enskip\enskip\enskip Greedily form $\pi$-series $s_\pi(x) = (X_i,X_{i_2}\ldots X_{i_N})$ starting from $X_i$ as its first \item[]\enskip\enskip\enskip element. 
      \enskip\enskip\enskip\hspace{5em} \For {each $\pi$-series $s_\pi(x)$ in step 3}\\ \enskip\enskip\enskip\enskip\enskip\enskip Find a corresponding smallest starting subset $\bf{T_x}$ with  $$M_\pi(\bf{T_x}) = \underset{ 1 \leq k \leq N-1}{\mathrm{max}} \pi(X_{i_{k+1}},\{X_{i_{1}},\ldots,X_{i_{k}}\})$$\EndFor \EndFor
    \State Among the non-coinciding minimal $\pi$-clusters $T_x$'s choose those that maximize $$ M_\pi(\bf{T_x})=\underset{X_i \in \mathbf{X}\setminus \mathbf{T_x}}{\text{min}}
  \pi(X_i,\mathbf{T_x}) $$
    \enskip\enskip\enskip all of which are the required minimal maximizers, and we return them as minimalMax\\

    \Return(minimalMax)
       \EndFunction

\end{algorithmic}
\end{algorithm}

The above algorithm finds all minimal maximizers in $\mathcal{O}(N^3g)$ time where $g$ is the average time required to compute the value of $\pi(X_i,\bf{S})$ for any $X_i,\bf{S}$. The fastest version of computing distance covariance to date is $\mathcal{O}(Nlog{}N)$ and proposed in \cite{FastDCoV}.
\begin{theorem}
 The algorithm above finds all the minimal maximizers over $\mathcal{P}^{-\bf{X}}$.
\end{theorem}

\begin{proof}
 From Theorem \ref{Theorem7.1} it follows that each element of minimalMax is a maximizer of $M_\pi(\bf{S})$ over $\mathcal{P}^{-\bf{X}}$.
 Assume that there is a minimal maximizer $\bf{S}$ that does not belong to minimalMax, and let $X_{i} \in \bf{S}$. Then, according to Theorem \ref{Theorem7.1}, there exist $\pi$-series starting from $X_i$ and minimal $\pi$-cluster $T_x \subseteq \bf{S}$ containing $X_{i}$ with $M_\pi(\bf{T_x}) \geq M_\pi(\bf{S})$.
 Since $\bf{S}$ does not belong to minimalMax, and, according to steps $5$ and $8$ of the algorithm, $T_x$ or some subset of $T_x$ belongs
 to minimalMax, there are a minimal maximizer strictly
 included in $\bf{S}$  which contradicts the minimality of $\bf{S}$.
\end{proof}

Putting all these results together we present our algorithm in Algorithm 1 above.

\section{All-Relevant Feature Selection}
In addition to our algorithm proposed above, we would like to point to a recent algorithm proposed in \cite{WahbaFS} for the purpose of solving Problem II using distance covariance. We present this algorithm in Algorithm 2 below.


 \begin{algorithm}[H]
   \caption{Kong-Wang-Wahba's All-Relevant Feature Selection algorithm\\ for Problem II:}
    \begin{algorithmic}[1]
      \Function{Kong-Wang-Wahba's Algorithm}{$\mathbf{X}$}
\State Calculate marginal sample distance correlations $\rho_n(X_i,Y)$ for variables $X_i$\item[]
\enskip\enskip\enskip\enskip for $i = 1, \ldots, n$ with the response $Y$.
\State Rank the variables in decreasing order of the sample distance correlations. Denote the ordered variables as $x_1, x_1, \ldots, x_n$. Start with $\bf{X_s} = \{x_1\}$.

\ForAll{$i$ from $2$ to $n$}\item[]\enskip\enskip\enskip\enskip Keep adding $x_i$ to $\bf{X_s}$ if $\nu_n(\bf{X_s}  Y)$, the
sample distance covariance, does not decrease.\item[]\enskip\enskip\enskip\enskip Stop otherwise.
\EndFor\item[]
    \Return($\bf{X_s}$)
       \EndFunction

\end{algorithmic}
\end{algorithm}

\section{Diverse and Relevant Feature Selection}

A methodological way of obtaining a solution for Problem III is by first running Kong-Wang-Wahba's All-relevant feature selection algorithm followed by running our proposed GreedyDiverseDCoV algorithm on the resulting solution of Kong-Wang-Wahba's algorithm. This would give a subset of the maximally separated diverse features that are also relevant with respect to the response. An alternate methodology would be to do the vice-versa of running our proposed GreedyDiverseDCoV algorithm first followed by running Kong-Wang-Wahba's algorithm on the resulting solution subset which is a union of the maximally separated subsets provided by our algorithm. This methodology of one before the other is analogous in principle
to the forward selection or backward selection methods for variable(feature) selection. That said this methodology of running both the GreedyDiverseDCoV and Kong-Wang-Wahba's algorithms  in series come with varied and useful theoretical guarantees as discussed in this paper and also deal with multiple objectives of diverse and relevant feature selection while also being completely model-free, free of distributional assumptions and being non-parametric.

\section{Experiments}
In this section we evaluate our above proposed combination of DiverseMinimalMaximDCoV Algorithm in Algorithm 1 for diverse selection applied on the subset returned by the relevant selection algorithm of Kong-Wang-Wahba in Algorithm 2. We compare this combination of diversity and relevancy encouraging feature selection with the mRMR Ensemble algorithm in \cite{mrmr} which also aims to select relevant and non-redundant (diverse) features.

\begin{figure*}[!htp]
\centering
\includegraphics[width=15cm]{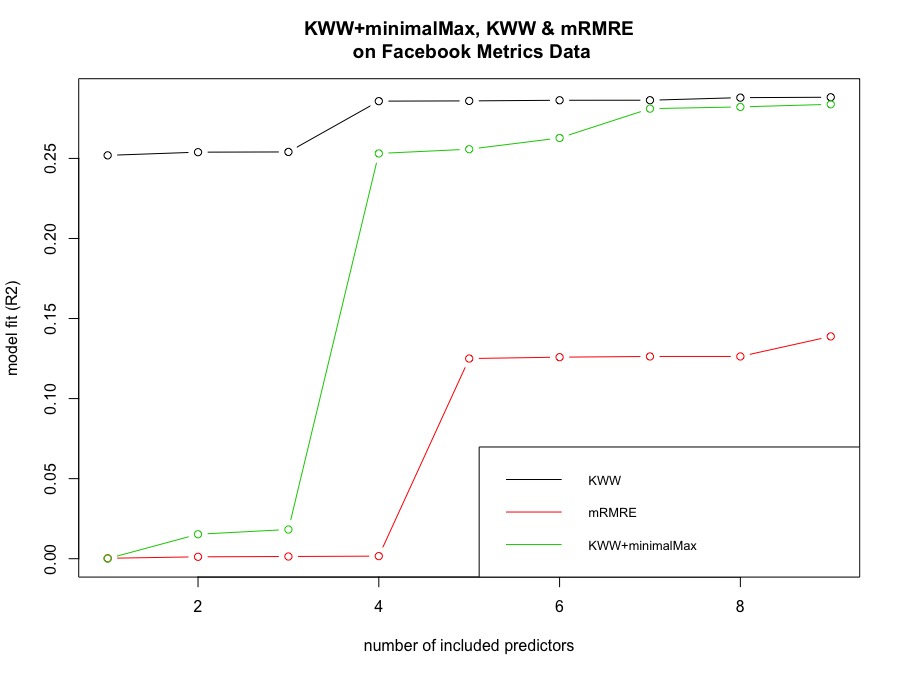}
\caption{Results on UCI's Facebook's comment volume prediction dataset, \url{https://archive.ics.uci.edu/ml/datasets/Facebook+Comment+Volume+Dataset}}
\label{fig:lion4}
\end{figure*}

\begin{figure*}[!htp]
\centering
\includegraphics[width=15cm]{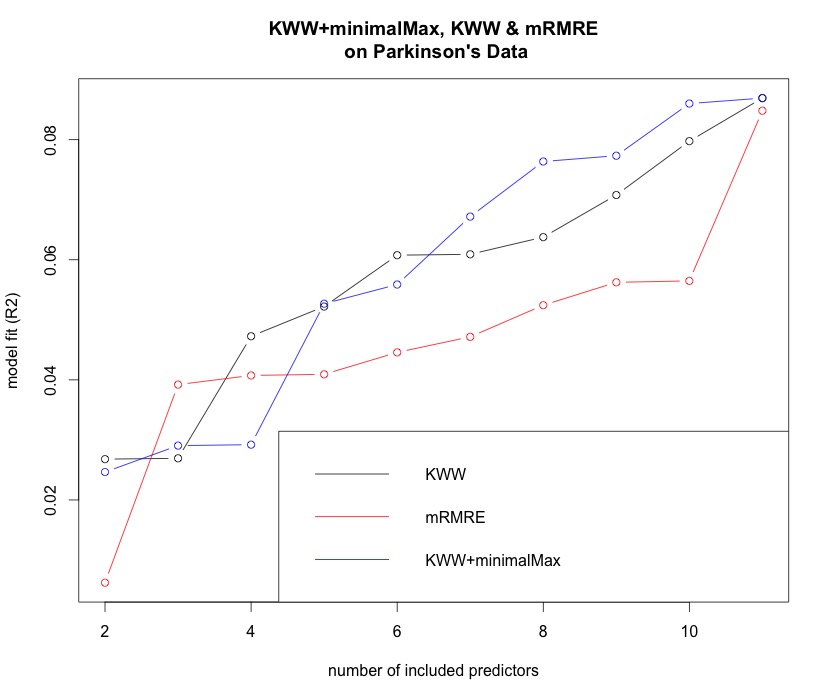}
\caption{Results on UCI's Parkinson Speech Dataset with Multiple Types of Sound Recordings,  \url{https://archive.ics.uci.edu/ml/datasets/Parkinson+Speech+Dataset+with++Multiple+Types+of+Sound+Recordings}}
\label{fig:lion1}
\end{figure*}

\begin{figure*}[!htp]
\centering
\includegraphics[width=15cm,height=15cm]{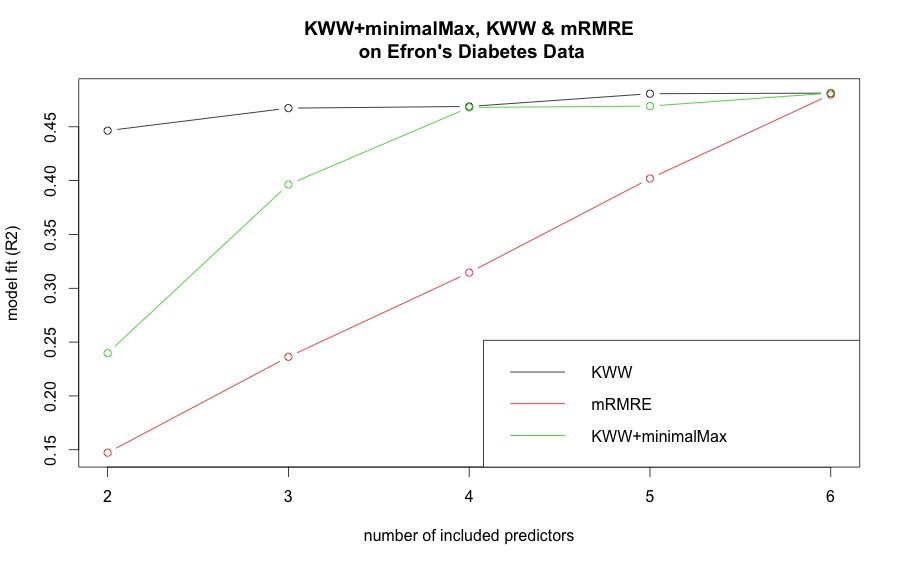}
\caption{Results on Diabetes data of 442 patients from Efron et al. 2004. Least angle regression, Annals of Statistics, 32:407-499, \url{http://artax.karlin.mff.cuni.cz/r-help/library/care/html/efron2004.html}}
\label{fig:lion2}
\end{figure*}

\begin{figure*}[!htp]
\centering
\includegraphics[width=15cm]{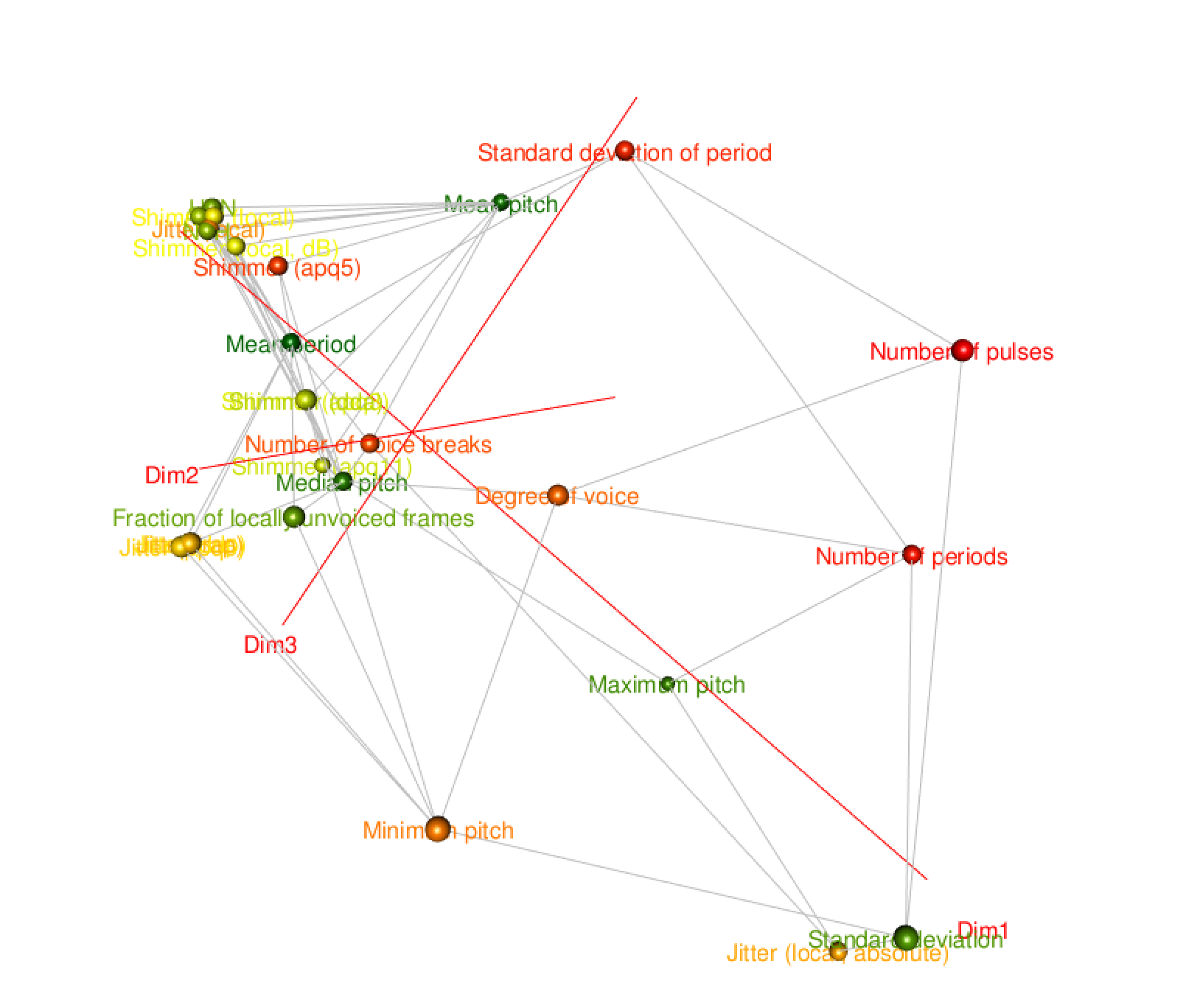}
\caption{ISOMAP Embedding:}
\label{fig:lion3}
\end{figure*}

\subsection{Datasets used in experiments:}
These are the three real-life datasets on which we evaluated the combination of Algorithm 1 on the results of Algorithm 2 and compared it with the mRMR Ensemble algorithm:
\begin{enumerate}
\item \textbf{UCI's Facebook's comment volume prediction dataset:} \\The goal associated with this dataset is to be able to predict the volume of comments on Facebook using various input metrics.
\item \textbf{UCI's Parkinson Speech Dataset with Multiple Types of Sound Recordings:} \\ We aimed to use speech data from Parkinson's patients with varying levels of severity and non-patients inorder to predict the UPDRS score, a score that is widely used in the medical fraternity to gauge the severity of Parkinson's in the subject under investigation.
\item \textbf{Efron's Diabetes data:}\\ This dataset consists of ten baseline variables of age, sex, body mass index, average blood pressure, and six blood serum measurements that were obtained for each of 442 diabetes patients, along with a response of interest, a quantitative measure of disease progression one year after baseline that was also collected. The goal associated is to be able to build a model that predicts this measure of disease progression.

\end{enumerate}

We present the results of this algorithmic comparison evaluated by the R-Squared Error metric upon fitting linear regression models on these three datasets in Figures 1, 2 and 3. As seen, our approach clearly outperformed the mRMR Ensemble model. In the case of our proposed Algorithm 1 we took an iterative approach of obtaining the minimal maximizer subsets and then noting them and removing them to regenerate minimal maximizers from remaining set of features. We continued this process till the end or till enough number of features were generated. We do note that the quality of the first iterate of minimal maximizers with regards to optimization of our proposed objective will be higher than subsequent subsets of minimal maximizers, but this is amongst the best one could do in our setting in order to generate an entire ordering of features from being most diverse with regards to our objective to the least diverse.  \par

In addition to this, it is interesting to analyze the gap between KWW+minimalMax line and KWW lines on the plot as it tells us how much the optimally relevant covariate selection of KWW matches with optimally diverse feature selection of our algorithm for any given dataset. So it tells us about the trade-off between relevancy and diversity of the covariates in any dataset like in a pareto frontier. Sometimes the relevancy maximizing and diversity maximizing subsets can intersect more and sometimes less based on the quality of the dataset in balancing these two criteria. Therefore this gap if quantified (say for example by integrating the difference between these lines) could be a good measure of evaluating the quality of any given dataset with respect to the relevancy-diversity tradeoff curve. This is just a direction we are pointing at and is not the main focus of our current paper. \par
In addition to this evaluation, we also performed a qualitative (and approximate) experiment to visually validate the diversity encouraging property of our theory. We did this by applying ISOMAP, a popular manifold learning technique on a matrix of pair-wise distance correlations between all pairs of features. This basically tries to generate a 2 dimensional Euclidean embedding like representation of of the Parkinson's dataset. This was presented in Figure 4 where we clearly were able to find the minimal maximizers produced by our algorithm to be farther from the rest of the features (as colored in red). We actually color coded the points from red to green in the order generated by our proposed algorithm 1. Therefore we would expect the red features to be more diverse than the green. Although this figure is an approximation of the behavior of features with regards to diversity, it still somewhat matches visually with exact solution of our formulation. 

In addition to comparisons with linear regression models on features selected by our approach and mRMR Ensemble, we also computed the 5 fold Cross- validated Mean Squared Error (MSE error) in predicting UPDRS scores with the Parkinson's dataset upon applying the random forest method of regression. Our combine approach of Algorithm 2 +  Algorithm 1 produces a lower MSE of 148.84 vs mRMRe which obtained 154.39 MSE. 

\section{An efficient pre-processing routine: The effect of scaling and centering on combinatorics of $\nu_n(\cdot)$ and $\rho_n(\cdot)$:}

We finally present an enumerative computational experiment we did to show that centering and scaling the data prior to applying our algorithms would lead to much better results as the distance covariances match up much better with distance correlations upon centering and scaling the data. This leads to the optimization of our proposed functions of distance covariance to auxiliarily mimic the optimization of our objective with distance correlation in the place of distance covariance. That is desirable as distance correlation is a normalized version of distance covariance. \label{enumExpt} \par
As part of these empirical enumerative experiments, we collected various popular real-life regression and classification datasets from the well known University of California-Irvine Machine Learning Repository (UCI-ML) and enumerated the entire power set of possible combinations of their features (covariates) $2^\mathbf{X}$. We then computed the distance correlations between each subset belonging to the power set and the response (or class-label) variable $\mathbf{Y}$. We denote these distance correlations by $\rho_{\bf{E}}$. We also computed the distance covariances between each subset belonging to the power set and the response (or class-label) variable $\mathbf{Y}$. We denote these distance covariances by $\nu_{\bf{E}}$ in the same arbitrary order of subsets used when computing $\rho_{\bf{E}}$. Now with this set of paired measurements of $\rho_{\bf{E}}, \nu_{\bf{E}}$ available across the entire power set of combinations of features we computed the distance correlation of  $\rho_{\bf{E}}, \nu_{\bf{E}}$ which we denote by  $\rho(\rho_{\bf{E}}, \nu_{\bf{E}})$ to see if combinatorially optimizing distance covariance over the power set is a good proxy (surrogate) for combinatorially optimizing distance correlation. The distance correlation $\rho(\rho_{\bf{E}}, \nu_{\bf{E}})$ happened to be very high in almost all cases and very close to the theoretical upper-bound of 1 which indicates a strong statistical dependence between $\rho_{\bf{E}}$ and $\nu_{\bf{E}}$, thereby directly pointing out to the fact that combinatorially optimizing $\nu_{\bf{S}}, S\subseteq 2^{\bf{X}}$ is a great proxy for combinatorially optimizing $\rho_{\bf{S}}$ over the power-set. We would also like to mention that the values were close to one in the case when the covariates(features or variables) were centered and scaled; an operation that is a widely accepted pre-processing for regression or classification modeling. We were motivated to contrast the highly-encouraging results produced after centering and scaling with respect to not performing a centering and scaling because of the fact that the sample distance correlation is a function of sample distance covariances and for a fixed response variable $\mathbf{Y}$, the numerator of sample distance correlation in equation 2 is dependent on both $\bf{X}$ and $Y$, while the denominator is only a function of $\bf{X}$ for a fixed response $\mathbf{Y}$. Thereby, the contribution of $\mid\mid \bf{X} \mid\mid$ on $\rho_n(\bf{X},Y)$ when $\rho_n(\bf{X},Y)$ can be reduced by scaling and centering the data prior to computing the distance covariance. This can be further motivated by the following identity that was proved in \cite{PVPaper}
\begin{align*}
    \nu_n(\bf{X,Y})=\Trace{\mathbf{X}^T\mathbf{L}_\mathbf{Y}\mathbf{X}}
    = &\frac{1}{2}\sum_{i,j=1}^{n}[\widehat{\mathbf{E}}_\mathbf{Y}]_{i,j}[\mathbf{E}_\mathbf{X}]_{i,j}.
\end{align*}
where $\widehat{\mathbf{E}}_\mathbf{Y}$ is the double-centered Euclidean distance matrix formed with the rows of $\mathbf{Y}$ being the points for computing the pair-wise distances on and $\mathbf{E}_\mathbf{Y}$ is the standard (without double-centering) Euclidean distance matrix of the rows of $\bf{X}$. This gives us that when $\bf{X}=Y$, the denominator of distance correlation is solely a function of $\bf{X}$ that can be standardized across $\bf{S}\in 2^{\bf{X}}$ by scaling and centering the values in $\bf{S}$. \par
All these results and comparisons of our enumerative experiment on the UCI-ML datasets are presented in Table 1 below.

\begin{table}[H]
\centering

\label{my-label}
\begin{tabular}{|l|l|l|l|l|}
\hline
\textbf{Dataset}                                 & \textbf{Dimensionality} & $\bf{\vert 2^{X} \vert - 1}$ & \begin{tabular}[c]{@{}l@{}}$\bf{\rho(\rho_{E}, \nu_{E})}$ \textbf{without}\\ \textbf{centering} \&\textbf{ scaling}\end{tabular} & \begin{tabular}[c]{@{}l@{}}$\bf{\rho(\rho_{E}, \nu_{E})}$ \textbf{with}\\ \textbf{centering} \&\textbf{ scaling}\end{tabular} \\ \hline
Airfoil Self-Noise                                                              & 1503 by 5      & 31                       & 0.896                                                                                           & 0.999                                                                                           \\ \hline
Abalone                                                                         & 4177 by 8      & 255                      & 0.422                                                                                           & 0.693                                                                                           \\ \hline
\begin{tabular}[c]{@{}l@{}}Banknote \\Authentication\end{tabular}                                                          & 1372 by 4      & 15                       & 0.938                                                                                           & 0.993                                                                                           \\ \hline
\begin{tabular}[c]{@{}l@{}}Concrete Compressive\\ Strength\end{tabular}         & 1030 by 8      & 255                      & 0.961                                                                                           & 0.965                                                                                           \\ \hline
\begin{tabular}[c]{@{}l@{}}Protein Localization \\ Sites of E.coli\end{tabular} & 336 by 7       & 127                      & 0.891                                                                                           & 0.966                                                                                           \\ \hline
Forest Fires                                                                    & 517 by 12      & 4095                     & 0.841                                                                                           & 0.941                                                                                           \\ \hline
Yacht Hydrodynamics                                                             & 308 by 6       & 63                       & 0.896                                                                                           & 0.999                                                                                           \\ \hline
\end{tabular}
\caption{A enumerative experiment with distance correlation and distance covariances over the power set }
\end{table}

\section{Conclusion:}
 We showed that our proposed Algorithm 1 gives exact solutions that are minimal-maximizers of our diversity encouraging objective. Similarly Algorithm 2 gives optimal solutions for a relevancy encouraging objective function. Now the quality of a solution subset that has a mixture of both properties of relevancy and diversity is dependent on pareto like trade-offs used in choosing the extent of diversity or relevancy one is willing to part away with unlike in highly optimal situations where the optimal solution of Algorithm 1 coincides with the optimal solution of Algorithm 2. That particular case would imply that the quality of the dataset being used for regression or classification is pretty optimal with regards to the relevancy-diversity tradeoff.

\end{document}